\documentclass[smallextended]{svjour3}       
\smartqed  

\usepackage{times}

\usepackage{url}
\usepackage{graphicx}        
\usepackage{subfigure}
\usepackage{amsmath}
\usepackage{amssymb}

\begin{document}
 
\journalname{Social Network Analysis and Mining} 
\title{Generalizing Unweighted Network Measures to Capture the Focus in Interactions\thanks{An earlier version of this paper was presented in the International Workshop on Social Network Analysis (SNAKDD) 2009.}}

\author{Sherief Abdallah}


\institute{S. Abdallah \at
              University of Edinburgh, UK \\
              British University in Dubai, UAE \\
              Tel.: +971-4-367-1964\\
              \email{shario@ieee.org}           
}

\date{Received:  ~~~~~~~~~~~~~~~~~~~~~~~~~~~~~~~~~~  / Accepted:   }

\maketitle 

\begin{abstract}
Unweighted network measures are commonly used to analyze real-world networks due to their simplicity and intuitiveness. This motivated the search for generalizations of unweighted network measures that take weights into account.
We propose a new generalization methodology that capture how focused are the interactions over edges. The less focused the interaction (more uniform over edges) the closer is our generalization to the original unweighted measure. None of the previously developed generalizations capture this aspect of weighted networks. We analyze several real world networks using our generalizations of the degree and the clustering coefficient. The analysis shows that our generalizations reveal interesting observations. 
\end{abstract}

\section{Introduction}

Mining and analyzing complex networks have received significant attention in recent years due to the explosive growth of social networks and the discovery of common patterns that govern wide-range of real world networks \cite{watts98n,barabasi99s,faloutsos99sigcomm,chakrabarti06acs,newman06phr,clauset07arxiv,park07pnas}. The core of mining complex networks is \emph{network measures}. Network measures are functions that summarize the network structure to simpler numeric values. These measures are generally classified into two main classes: measures that ignore edge weights and focus primarily on the structure of the graph, which we call unweighted measures, and measures that take edge weights into account (in addition to the structure), which we call weighted measures.

Unweighted measures received the bulk of researchers' attention, due to their simplicity, intuitiveness, and the relative ease of computation. Such an attention resulted in several influential findings such as the small world (relied on the clustering coefficient) \cite{watts98n} and the power-law (relied on the degree distribution)\cite{barabasi99s,clauset07arxiv}. 
Despite their popularity, unweighted measures ignore important network information: the weights. Consequently, several measures were developed in order to take weights into account. The use of weighted measures, however, is still dwarfed by the use of unweighted measures in analyzing complex networks \cite{boccaletti06pr,leskovec07atkdd,clauset07arxiv,park07pnas}. 

The wide spread usage of unweighted network measures motivated the search for generalizations of unweighted measures that takes weights into account \cite{barrat04pnas,ahnert07pre,opsahl09sn,opsahl10sn}. 
For example, one generalization of the degree computed the expected number of edges incident to a node, assuming the weight of an edge reflects the probability of the edge existing \cite{ahnert07pre}. The most recent generalization of the degree measure used a tunable parameter $\alpha$ to mix both the unweighted degree and the strength using simple multiplication \cite{opsahl10sn}.\footnote{A node's degree is the number of edges incident to the node, while a node's strength is the summation of weights incident to the node. Section \ref{sec-effcar} provides the formal definitions.} 
Section \ref{sec-related} describes in more detail the previous generalizations and the related work. 



We propose in this paper a new generalization methodology that captures the degree of \emph{focus} of interaction. If the interaction is not focused and uniform over edges (uniform weights), then our generalized measure reduces to the original unweighted measure (and this reduction is guaranteed). So for example, if a node has five neighbors and it interacts with all five neighbors equally (no focus), then our generalization of the degree reports the node to have a generalized degree of 5, similar to the unweighted degree. If the node focuses and limits its interaction with only two out of five neighbors, then our generalization of the degree will capture this focus and report the node to have a generalized degree closer to 2 rather than 5.




Consider the simplified scenario in Figure \ref{fig-example-evolution}[a] for illustration. Four students met for the first time in some class. Initially all of the four students are interacting uniformly with one another, forming a clique with equal edge weights. The weight attached to a link (edge) between two students quantifies the amount of interaction between the two students. As time passes, each student focuses her interaction on fewer subset of students (friends). As a result, some links get weaker and eventually disappear, while other links get stronger. The final interaction network shows lesser average degree. 

    \begin{figure}
    \centering
		\includegraphics[width=4.2in]{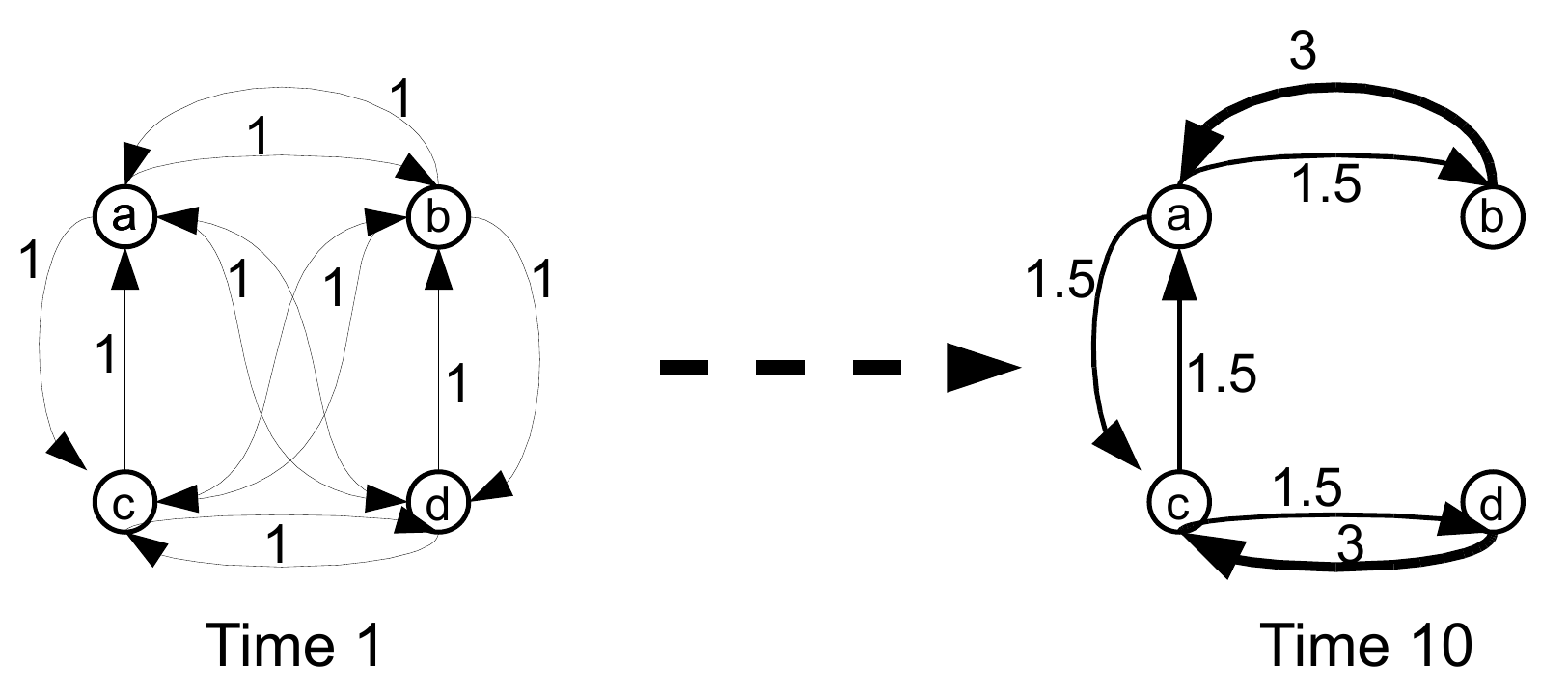}
      \caption{\small A simplified scenario of a dynamic network. The figure shows the evolution of interaction among 4 students in a class. In the beginning (Time 1) every student interacts with every other student equally (every weight equals 1). As time passes, some links get weaker while other links get stronger. At the end (Time 10), only subset of links have weight greater than 0, and not all the weights are equal.}
      \label{fig-example-evolution}
    \end{figure}

Now suppose we monitor the evolution of edge weights over time at different snapshots. For simplicity, suppose the weights of different edges increase (or decrease) linearly from Time 1 to Time 10. So for example, the edge from node $a$ to node $b$ increases by a rate of $0.05$ per time unit. Similarly, the edge from node $c$ to node $b$ decreases by a rate of $-0.1$ per time unit. Notice here that the amount of out-ward interaction of each node (the strength) remains constant over time (and equal 3). However, out-ward interaction becomes focused towards 2 or less neighbors. Figure \ref{fig-example-evolution-hist} shows the corresponding evolution of the degree distribution, the strength distribution, and the $\alpha$-degree distribution for two values of $\alpha$: 0.5 and 1.5 \cite{opsahl10sn}.\footnote{We use the generalized $\alpha$-degree as a representative of the state-of-the-art generalizations \cite{opsahl10sn}. The $\alpha$ values of 0.5 and 1.5 were proposed by the original paper.} Notice that neither the strength nor the $\alpha$-degree reflects the degree distribution of the final network at Time 10. Furthermore, none of the measures capture the change in the \emph{focus} of interaction over time. In other words, although the change in the interaction pattern between students was gradual, this gradual change is not captured by any of the measures (instead, there is a sudden jump in the distribution).

    \begin{figure}
    \centering
		\subfigure[Degree Distribution]{\includegraphics[width=2.25in]{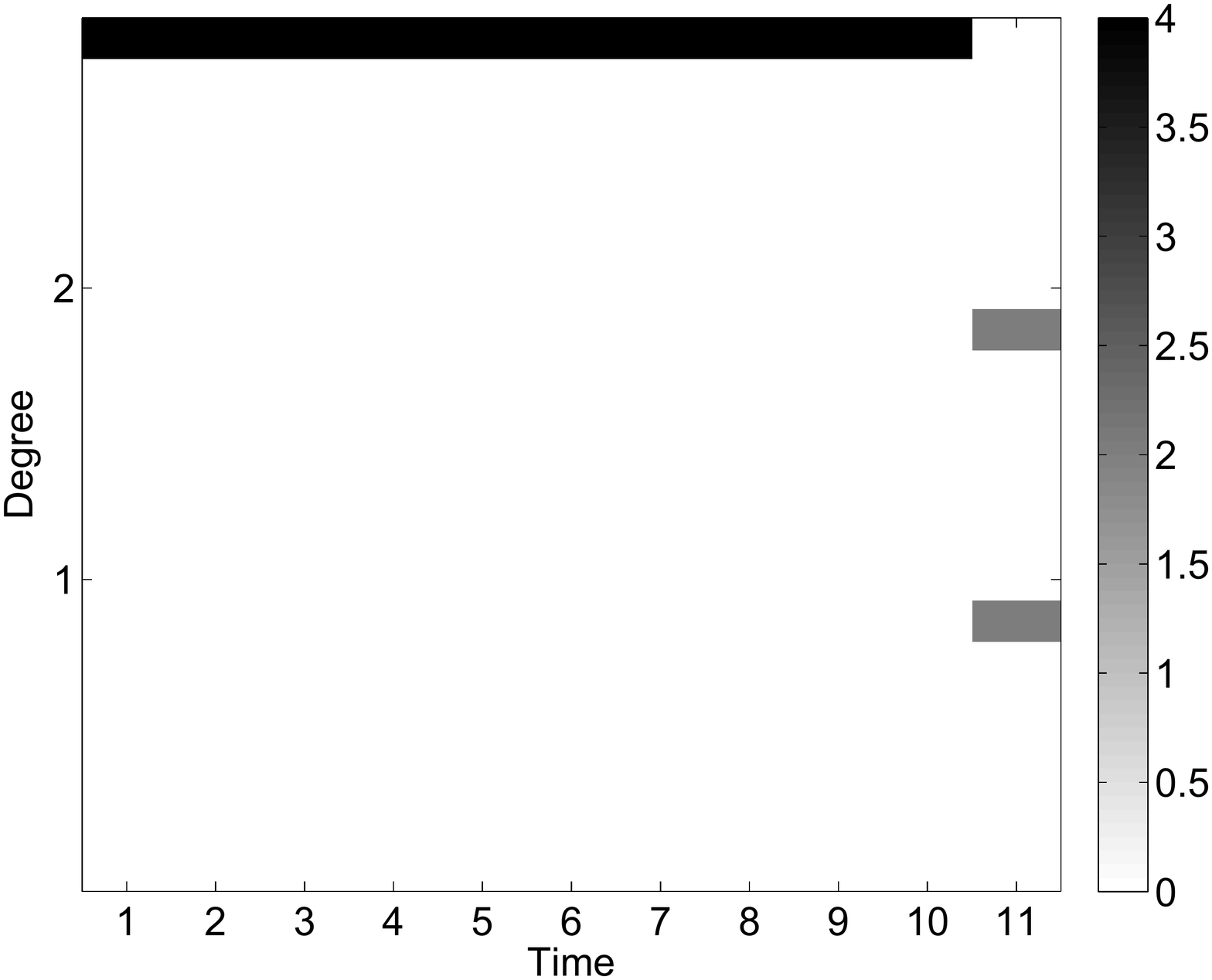}}
		\subfigure[Strength Distribution]{\includegraphics[width=2.25in]{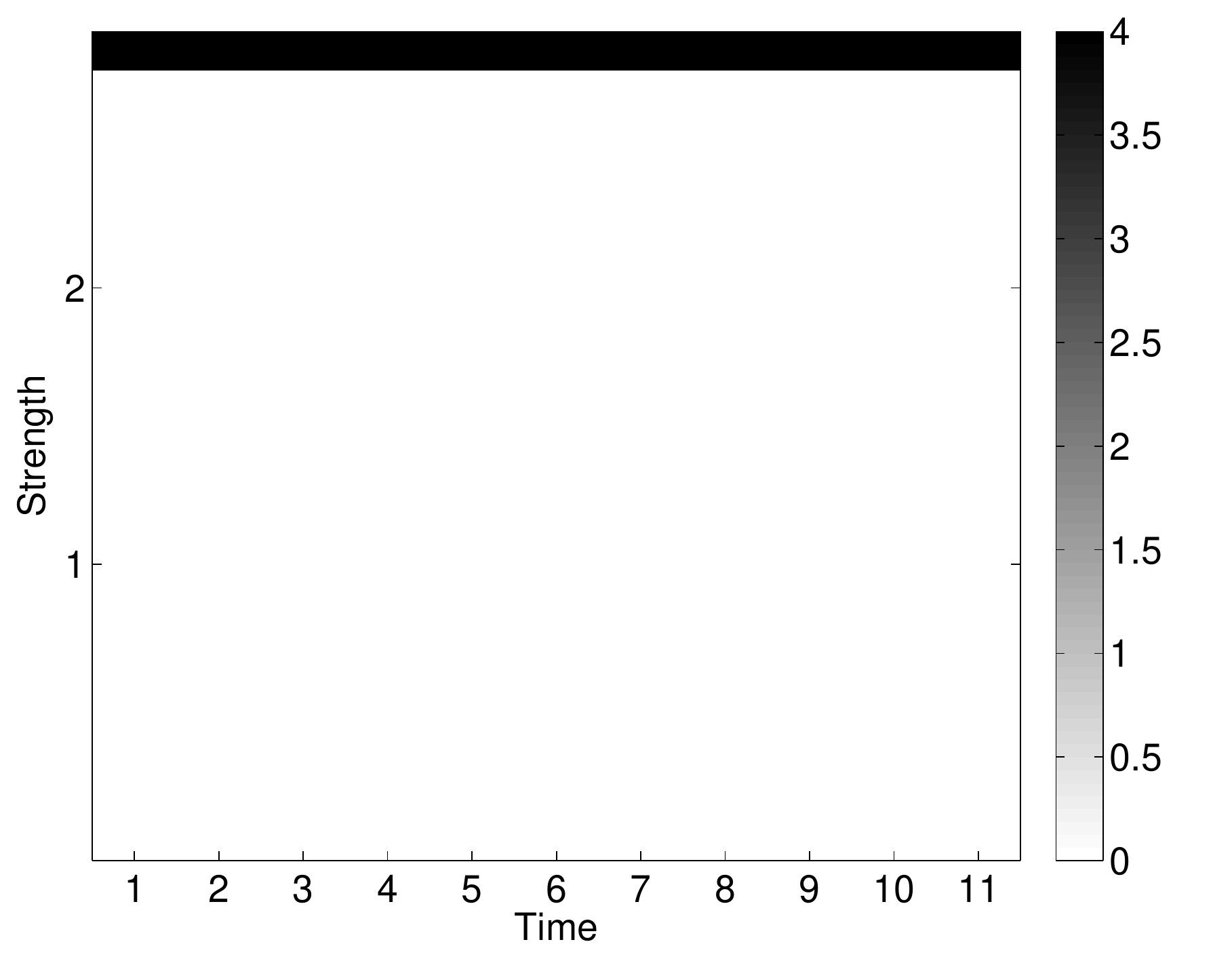}}
		\subfigure[$\alpha$-Degree Distribution, $\alpha=0.5$]{\includegraphics[width=2.25in]{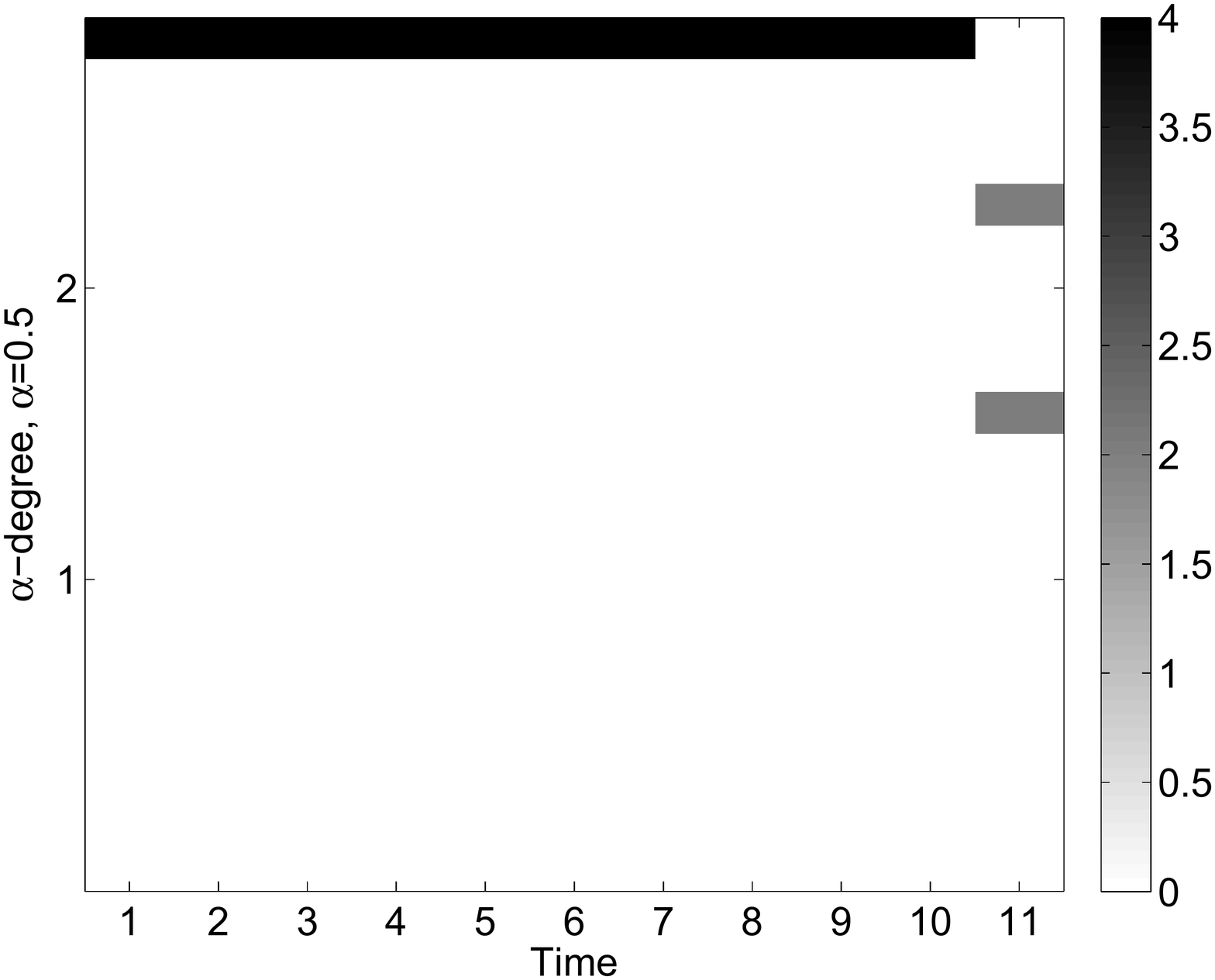}}
		\subfigure[$\alpha$-Degree Distribution, $\alpha=1.5$]{\includegraphics[width=2.25in]{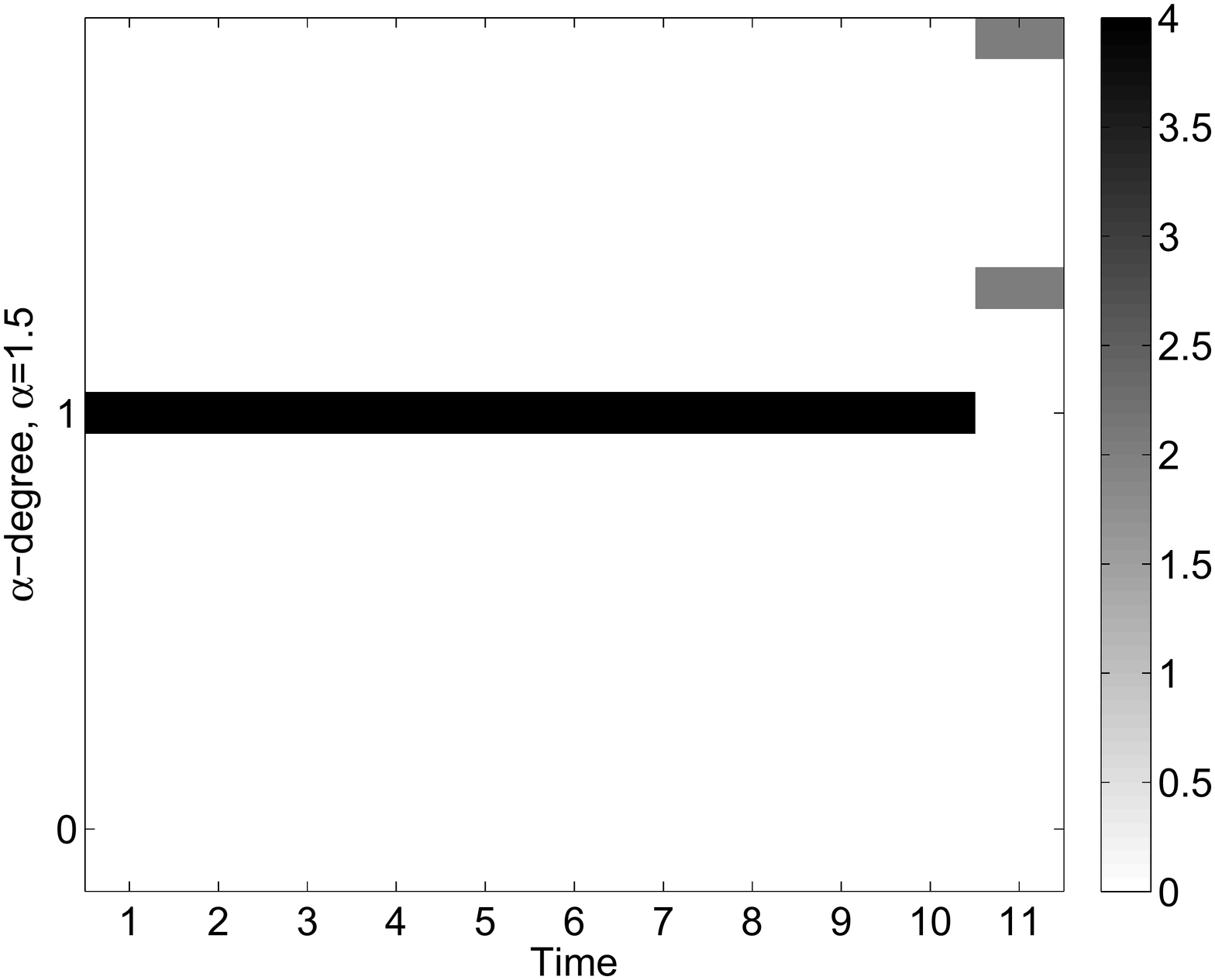}}
		\caption{\small The corresponding evolution of common network measures for the students dynamic network: the degree distribution, the strength distribution, and the $\alpha-$degree distribution for $\alpha =0.5$ and $\alpha=1.5$. Two issues here are worth noting: (1) neither the strength nor the $\alpha$-degree reflects the degree distribution of the final network at Time 10 (2) None of the measures capture the change of interaction \emph{focus} in a gradual continuous manner.}
      \label{fig-example-evolution-hist}
    \end{figure}

Figure \ref{fig-example-evolution-cd} shows the evolution of our degree generalization for the student network scenario. Notice here the continuity of our measure (Figure \ref{fig-example-evolution-cd}[b]), in contrast to the unweighted degree. Notice also that when there at Time 1 and Time 10, our generalized degree is identical to the original unweighted degree.

    \begin{figure}
    \centering
		\subfigure[C-Degree Distribution]{\includegraphics[width=2.25in]{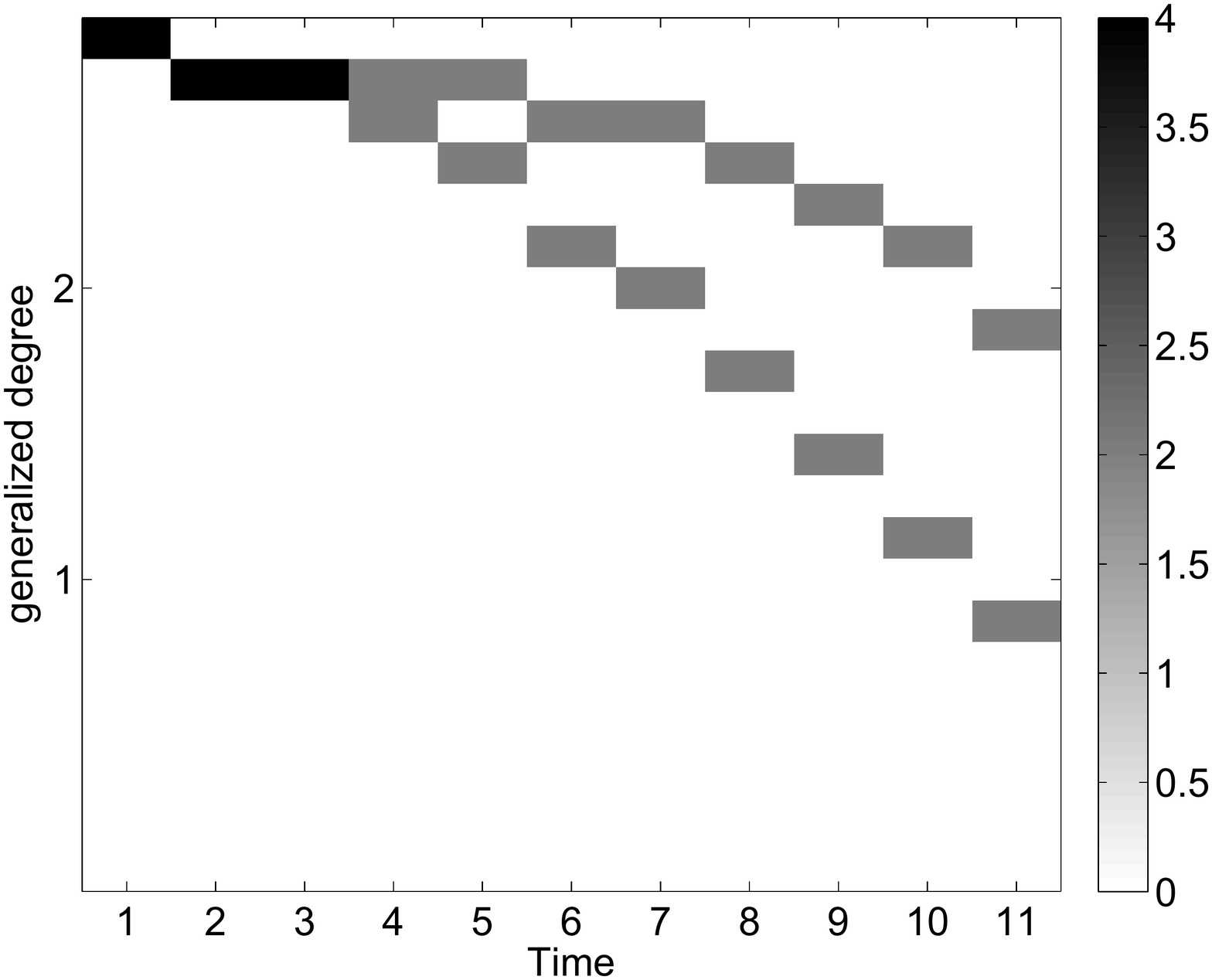}}
		\subfigure[C-degree evolution]{\includegraphics[width=2.25in]{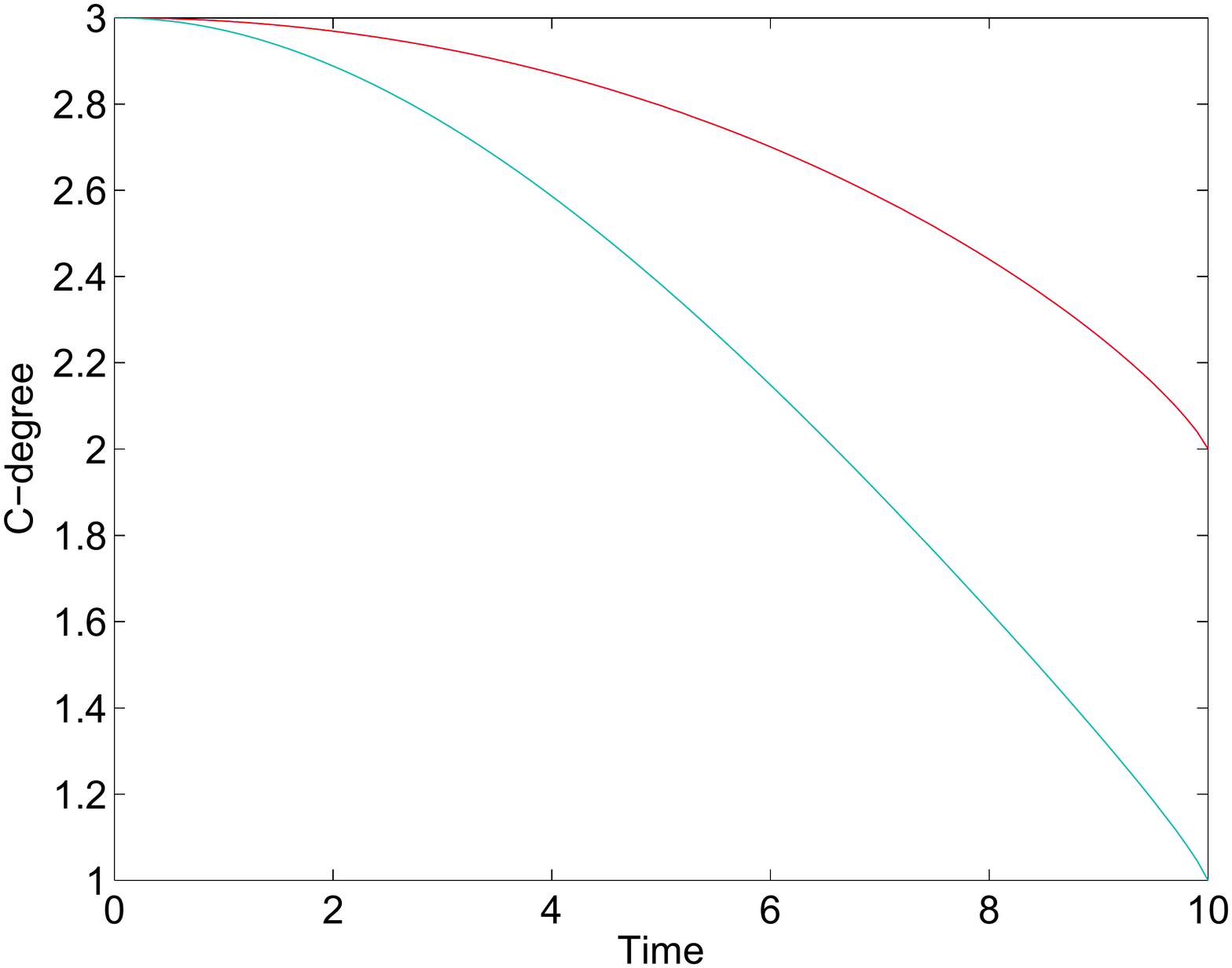}}
		\caption{\small The corresponding evolution of our generalization of the degree measure, the continuous-degree (C-degree) for the students dynamic network: the degree distribution. Part (a) shows the histogram of the C-degree over time, discretized into 20 bins. Part(b) shows the C-degree for the four nodes over finer time-scale. Notice the continuous evolution of our generalization and the direct connection to the unweighted degree at the beginning and at the end of the timescale.}
      \label{fig-example-evolution-cd}
    \end{figure}


We illustrate the applicability of our method by generalizing four unweighted measures: the node degree, the clustering coefficient, the dyadicity, and the heterophilicity.
As a case study, we analyze several real-world, weighted, social networks using two generalized measures: the C-degree and the C-clustering coefficient (the letter C stands for \emph{continuous} and denotes our generalization of an unweighted measure). 


But before we describe our contribution in the following section, it is important to state the limitations and the assumptions of our approach. We assume that an edge's weight reflects the amount of interaction across that edge, which is a reasonable assumption in most real-world domains. For example, an edge weight can represent the number of times a person calls a friend, the number of emails transmitted to an email address, or the number of papers co-authored by two scientists. On the other hand, if weights reflect something like the dissimilarity between neighbors, then our approach is not suitable.
More importantly, our generalization captures the focus of interactions \emph{not the intensity} of the interactions. For example, suppose that in the previous student network scenario we multiply the weights (over all edges) by some constant every time step (instead of adding or subtracting). In this case, both the degree and the C-degree (our generalization)  will not change, but the strength and the $\alpha$-degree will change.

The following section describes the heart of our approach: generalizing the cardinality concept of a set to take weights into account.


\section{Generalizing Measures Using Generalized Cardinality}
\label{sec-effcar}

Several unweighted measures use the cardinality (the size) of some subset of edges in their computation. For example, the node degree is the number of edges incident to a node. The clustering coefficient of a node is the ratio between the number of edges between its neighbors and the number of all possible edges among the neighbors.\footnote{Other examples include heterophilicity and dyadicity. We describe these measures in further detail later.}
The main limitation of the traditional cardinality function (and consequently all the unweighted network measures that use it) is that it ignores edge weights. We show in this section how to generalize the cardinality to capture the focus in interactions. In the following section we show how to use the generalized cardinality to generalize the degree, the clustering coefficient and other unweighted measures that use the cardinality of some edge set.

To put it more formally, let $E'=\{e_1, ..., e_n\} \subseteq E$ be the subset of edges that are used in computing a particular network measure, where $E$ is the set of all network edges and the cardinality of $E'=n=|E'|$. The degree of node $i$ is then defined as $k(i)=|E'_i|$, where $E'_i$ is the set of edges incident to node $i$. Similarly, the clustering coefficient of node $i$ is $z(i)=\frac{|E'_i|}{MAX^N_i}$, where $E'_i$ here is the set of edges between node $i$'s neighbors and $MAX^N_i$ is a constant that equals the maximum number of edges that can exist between these neighbors (i.e. if node $i$'s neighbors formed a clique). 
In weighted networks, each edge  $e \in E$ has a corresponding non-negative weight $w(e) \ge 0$. 
The cardinality implicitly assumes uniform weights over the edges. When weights are not uniform, the cardinality can give an incorrect perception of the actual use of edges. 
Consider the following numeric example. There are four sets of edges with corresponding multisets of weights $W_1 = \{5,5,5,5\}, W_2=\{9,5,5,1\}, W_3=\{9,8,2,1\}$ and $W_4=\{20,0,0,0\}$. The cardinalities of these weight multisets are all the same and equals 4. Intuitively, however, if the weights reflect the interaction over edges, then not all the edges are being used equally and the traditional cardinality becomes a crude approximation. It is possible to define a cutoff threshold weight (an edge is included in the graph if its weight is above a threshold, otherwise the edge is excluded). The computation of any unweighted measure then takes place naturally \cite{chapanond05cmot,kalisky06phr}. Such an approach, however, does not properly handle the focus of interaction among neighbors, and it is not clear how big or small should the threshold be. 

Instead, we want a \emph{generalized cardinality} function that summarizes a set of weighted edges into a single real number and has two properties. If edge weights are equal (no focus), then the function we are looking for should return the traditional cardinality. When the weights are not equal, the function should assign a value between 1 and the traditional cardinality such that the more equal the weights are (less focus), the higher the value. The two properties ensure the consistent connection between the original measure and the generalization. The second property ensures that the generalized measure captures the degree of focus in interactions. Using the previous numeric example, we are looking for a function that assigns 4 to $W_1$, 1 to $W_4$, and values between 1 and 4 for $W_2$ and $W_3$, with the value assigned to $W_2$ greater than the value assigned to $W_3$ (because the two inner weights are equal in case of $W_2$). Such a generalization of the cardinality measure will allow straightforward generalization of many unweighted network measures.
Simple functions for summarizing sets (such as the average, the variance, and the summation) can be very useful in summarizing weights, but they do not satisfy the two desired properties mentioned above. 

The heart of our generalization is a generalized definition of the cardinality of a set of edges $E'$ that takes weights into account, which we call the \emph{effective cardinality}, or $c(E')$:

\[
c(E')= \left\{ \begin{array}{ll}
			0 & \textrm{if $E'$ is empty} \\
			2^{\left(\sum_{e\in E'} \frac{w(e)}{\sum_{o \in E'} w(o)} \log_2 \frac{\sum_{o \in E'} w(o)}{w(e)}\right)} & otherwise
			\end{array} \right.
\]

Intuitively, the quantity $\frac{w(e)}{\sum_{o \in E'}w(o)}$ represents the probability of an interaction over an edge $e$ among all the edges in $E'$. The multiset $\left\{\frac{w(e)}{\sum_{o \in E'}w(o)}: e \in E' \right\}$ is a probability distribution over edges and $H(E')=\sum_{e\in E'}\left[ \frac{w(e)}{\sum_{o \in E'}w(o)} \log_2 \frac{\sum_{o \in E'}w(o)}{w(e)} \right]$ is the entropy of this probability distribution. The entropy measures the disparity between the weights: the more uniform the weights are, the higher the entropy and vice versa.\footnote{Note that the quantity $x \log_2 \frac{1}{x} \rightarrow 0$ as $x \rightarrow 0$ or $x=1$.} The purpose of the power 2 is to convert the entropy back to the number of edges that are effectively being used. In other words, the effective cardinality of edge set $E'$ returns the number of edges of equal weights that has the same entropy as the edges in $E'$.

Before discussing the important properties of the effective cardinality, let us first consider few numeric examples that illustrate the intuition. Consider the multiset\footnote{because more than one edge can have the same weight.} of weights $W=\{10,0.01\}$. The traditional cardinality of this set is 2. However, if weights quantify the amount of interaction over edges, then the interaction is highly focused over one edge and the cardinality should be closer to 1 than 2.  For $W=\{10,0.01\}$, $c(\{10,0.01\})=1.008$, so even though the set $W$ has two edges, the effective cardinality is equivalent to only $1.008$. 
For the numeric example we have mentioned earlier, $c(W_1 = \{5,5,5,5\})=4$, $c(W_2=\{9,5,5,1\})=3.3276$, $c(W_3=\{9,8,2,1\})=3.0219$ and $W_4=\{20,0,0,0\} = 1$, which satisfy the intuitive ordering we described earlier. This ordering is not just by chance or due to a special case, but is actually guaranteed by our proposed effective cardinality.

The effective cardinality satisfies three properties (proofs follow from entropy properties and are given in the appendix):

\begin{enumerate}
\item \textbf{Preserving maximum cardinality:} $\forall E': c(E') \le |E'|$. Furthermore, $c(E') = |E'|$ iff $\forall e \in E': w(e)=C$, where $C$ is some constant. In other words, the effective cardinality is maximum and equals the original cardinality when there is no disparity between weights.  
\item \textbf{Preserving minimum cardinality:} $c(E')=0$ iff $E'$ is an empty set. Furthermore, $c(E')=1$ iff $\exists u \in E': w(u) > 0$ and $\forall v \ne u: w(v)= 0$. In other words, the effective cardinality is one when all edges, except one edge, have zero weights. 
\item \textbf{Consistent partial order over weighted sets:} any function that maps a  set of real numbers (weights) to a single real number imposes an implicit partial order. The effective cardinality imposes, arguably, the simplest partial order that is consistent with the above two properties. If the two sets of weighted edges have the same size, the same summation of weights, and their individual weights are the same except for two edges, then the set with more uniform  weights has higher effective cardinality (a formal definition of this property is given in Lemma \ref{lemma-po} in the appendix). 
\end{enumerate}

The three properties ensure a consistent connection to the original cardinality. The properties also confirm that the effective cardinality (and consequently any generalized measure based on it) captures and is sensitive to the focus of interaction (the disparity between weights).
The effective cardinality, however, is \emph{not sensitive to the scale} of weights. So for example, the multiset of weights $\{1,1\}$ has the same effective cardinality as the multiset $\{10,10\}$. This can be contrasted to the traditional degree, which is sensitive to neither the scale nor the disparity, and the strength, which is not sensitive to the disparity of weights but is sensitive to the scale. We come back to this issue in our analysis in Section \ref{sec-experiment}.
The following section uses the effective cardinality to generalize some unweighted measures. 


\section{Generalizing Unweighted Network Measures Using Effective Cardinality}
\label{sec-general}

In principal, unweighted network measures which use the cardinality of some subset of edges can be generalized using the effective cardinality. In fact, while we limited the discussion in the previous section to sets of weighted edges, all discussed properties apply to any multiset of weights, \emph{even if elements in the set represent subgraphs, not edges}. So for example, if we are interested in counting triangles of three connected vertices (which are used in some definitions of the clustering coefficient), we can use the effective cardinality to replace the discrete count of triangles with a continuous spectrum. 
This section presents four example generalizations of unweighted network measures: the degree, the clustering coefficient, the dyadicity, and the heterophilicity. The resulting generalized measures inherit the three properties of the effective cardinality. Table \ref{tab-summary} summarize these generalizations.

\begin{table}[htbp]
\centering
        \begin{center}
        \begin{tabular}{|l|l|l|l|}
        \hline 
        Measure & Unweighted & Generalized \\
        \hline 
        Degree of node $i$& $|E_i|$ & $c(E_i)$\\
        \hline 
        Clustering coefficient of node $i$& $\frac{|E^N_i|}{MAX^N_i}$ & $\frac{c(E^N_i)}{MAX^N_i}$ \\
        \hline 
        Dyadicity of a graph & $\frac{|E_{within}|}{n_{within}}$ & $\frac{c(E_{within})}{n_{within}}$ \\
        \hline
        Heterophilicity of a graph & $\frac{|E_{across}|}{n_{across}}$ & $\frac{c(E_{across})}{n_{across}}$ \\
        \hline
\end{tabular}
\end{center}
\caption{\small{The summary of the generalization of four unweighted measures, where $E_i$ is the set of edges incident to node $i$, $E^N_i$ is the set of edges between neighbors of node $i$, $E_{within}$ is the set of edges within a class of nodes, and $E_{across}$ is the set of edges across two classes of nodes.}}
\label{tab-summary}
\end{table}

The dyadicity and heterophilicity were recently used to study the correlation between the types of nodes (node classes) and the network structure \cite{park07pnas}. The dyadicity of a graph equals $\frac{|E_{within}|}{n_{within}}$, where $E_{within}$ is the set of edges within a set of nodes of the same type (a class of nodes) and $n_{within}$ is the expected number of edges within the same class of nodes if there was no correlation between the node class and the network structure. Intuitively, the dyadicity quantifies the strength of connections between nodes of the same type and whether it is above average.\footnote{There are other network measures that also quantified the strength of connections within a class (community) of nodes, such as the modularity measure \cite{newman04preB}.} The heterophilicity of a graph equals $\frac{|E_{across}|}{n_{across}}$, where $E_{across}$ is the set of edges across two classes of nodes and $n_{across}$ is the expected number of edges across the two classes if there was no correlation between the node class and the network structure. The heterophilicity quantifies the strength of connections across two classes (communities) of nodes and whether it is above average.
The dyadicity can be generalized, using the effective cardinality, to be $\frac{c(E_{within})}{n_{within}}$ and similarly the heterophilicity can be generalized to be $\frac{c(E_{across})}{n_{across}}$. 

The degree and the clustering coefficient, of a particular node, are two of the most widely used unweighted measures, so the remainder of this section focuses on their generalization.


\subsection{Generalizing the Degree}
\label{sec-case}

A node's degree is the number of edges incident to the node, or $|E_i|$, where $E_i$ is the set of edges incident to node $i$. 
The degree distribution (the histogram of the degrees of all network nodes), has been used extensively to analyze and characterize networks, and helped in discovering common patterns in real world networks, particularly the power law \cite{barabasi99s,faloutsos99sigcomm,boccaletti06pr,clauset07arxiv,leskovec07atkdd}. A degree distribution follows the power law if $P(k) \propto k^{-\alpha}$, where $k$ is the degree, $\alpha$ is a constant, and $P(k)$ is the degree distribution. 
A generalization of the degree measure using the effective cardinality, which we call the continuous degree (C-degree), is straightforward:

\begin{definition}
The \emph{C-degree} of a node $i$ in a network is $r(i)$, where
\[ 
r(i) = c(E_i) = \left\{ \begin{array}{ll}
			0 & \textrm{if $i$ is disconnected} \\
			2^{\left(\sum_{e\in E_i} \frac{w(e)}{s(i)} \log_2 \frac{s(i)}{w(e)}\right)} & otherwise
			\end{array} \right.
\]
\end{definition}

The set $E_i$ is the set of edges incident to node $i$ and $s(i)=\sum_{e \in E_i}w(e)$ is the strength of node $i$.
Figure \ref{fig-cdd} compares the continuous degree distribution to the (discrete) degree distribution in a simple weighted network of four nodes. A node on the boundary has an out degree of 1, while an internal node has an out degree of 2. Intuitively, however, only  one of the internal nodes is fully utilizing its degree of 2 (the one to the left), while the other node (to the right) is mostly using one neighbor only. The C-degree measure captures this and shows that the internal node to the left has a C-degree of $c(\{0.5,0.5\})=2$ while the other internal node has a C-degree of $c(\{0.9,0.1\})=2^{H(0.9,0.1)}=1.38$. 

    \begin{figure}
    \centering
\includegraphics[width=4in]{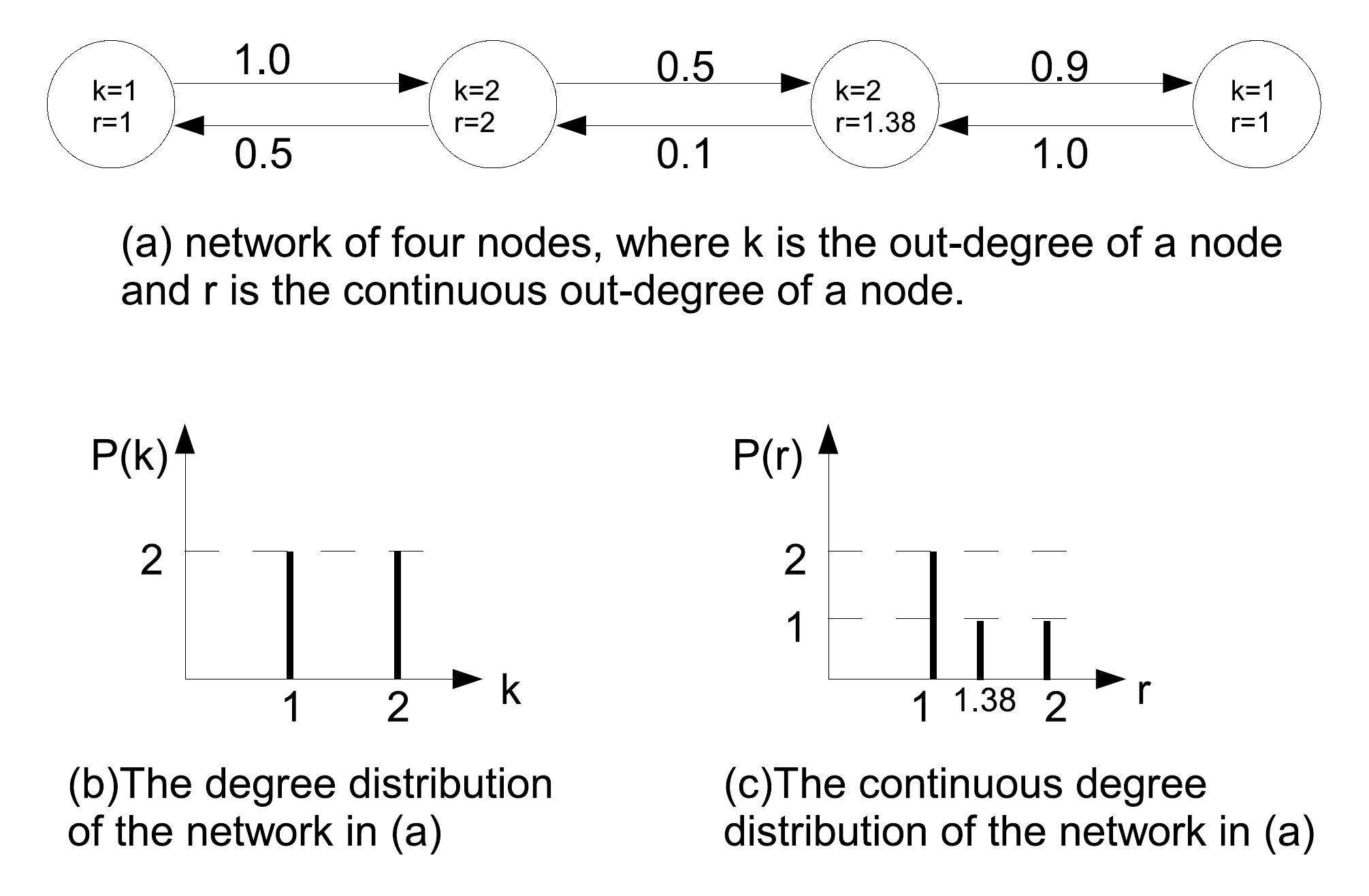}
      \caption{\small Example weighted network of four nodes, comparing the (discrete) degree against the C-degree. The degree distribution illustrates the benefit of taking weights into account in distinguishing nodes.}
      \label{fig-cdd}
    \end{figure}

The C-degree inherits the three properties we described earlier with respect to the traditional node degree. The C-degree of a node is maximum and equals the traditional discrete degree when all the weights incident to the node are equal. The C-degree of a connected node is minimum and equals one if all edges incident to the node have zero weights except one edge that has a weight greater than zero. And finally, everything else being equal, a node with more uniform weights incident to it (less focused interaction) has higher C-degree than a node with less uniform weights incident to it. 

\subsection{Generalizing the Clustering Coefficient}
\label{sec-cluster}

The clustering coefficient is a measure that quantifies the clustering or connectivity among a node's neighbors. 
When averaged over all nodes, the clustering coefficient represents the connectivity of the whole network. The clustering coefficient is an important property for identifying small world networks \cite{watts98n} and is given by the equation $\frac{|E_i^N|}{MAX_i^N}$, where $E_i^N$ is the set of edges between node $i$'s neighbors and $MAX^N_i$ is the maximum number of edges that can be between these neighbors.\footnote{Note that, particularly for directed graphs, some researchers argued that a \emph{clustering signature} would be more suitable in distinguishing networks \cite{ahnert08pre}. In a clustering signature, 7 types of directed triangles are counted separately. The effective cardinality can still be used to replace the discrete counts of these triangles. For the purpose of this paper we focus on the simpler, more widely used definition of the clustering coefficient.}  The generalized clustering coefficient of a node $i$ using the effective cardinality is:

\[ 
o(i)=\frac{c(E_i^N)}{MAX_i^N}
\]

Figure \ref{fig-example-cco} provides a simple motivating example of 3-nodes. The C-clustering coefficient can help in distinguishing different nodes that are deemed indistinguishable using the traditional clustering coefficient. For example, both nodes $A$ and $B$ have a clustering coefficient of $2/2=1$ (neighboring nodes  have two edges between them, out of two possible edges). However, $B$'s C-clustering coefficient is $o(B)=c(\{5,1\})/2=0.78$, while the C-clustering coefficient of $A$ is $o(A)=c(\{5,5\})/2 = 1$.

\begin{figure}[htbp]
\centering
  \includegraphics[width=4in]{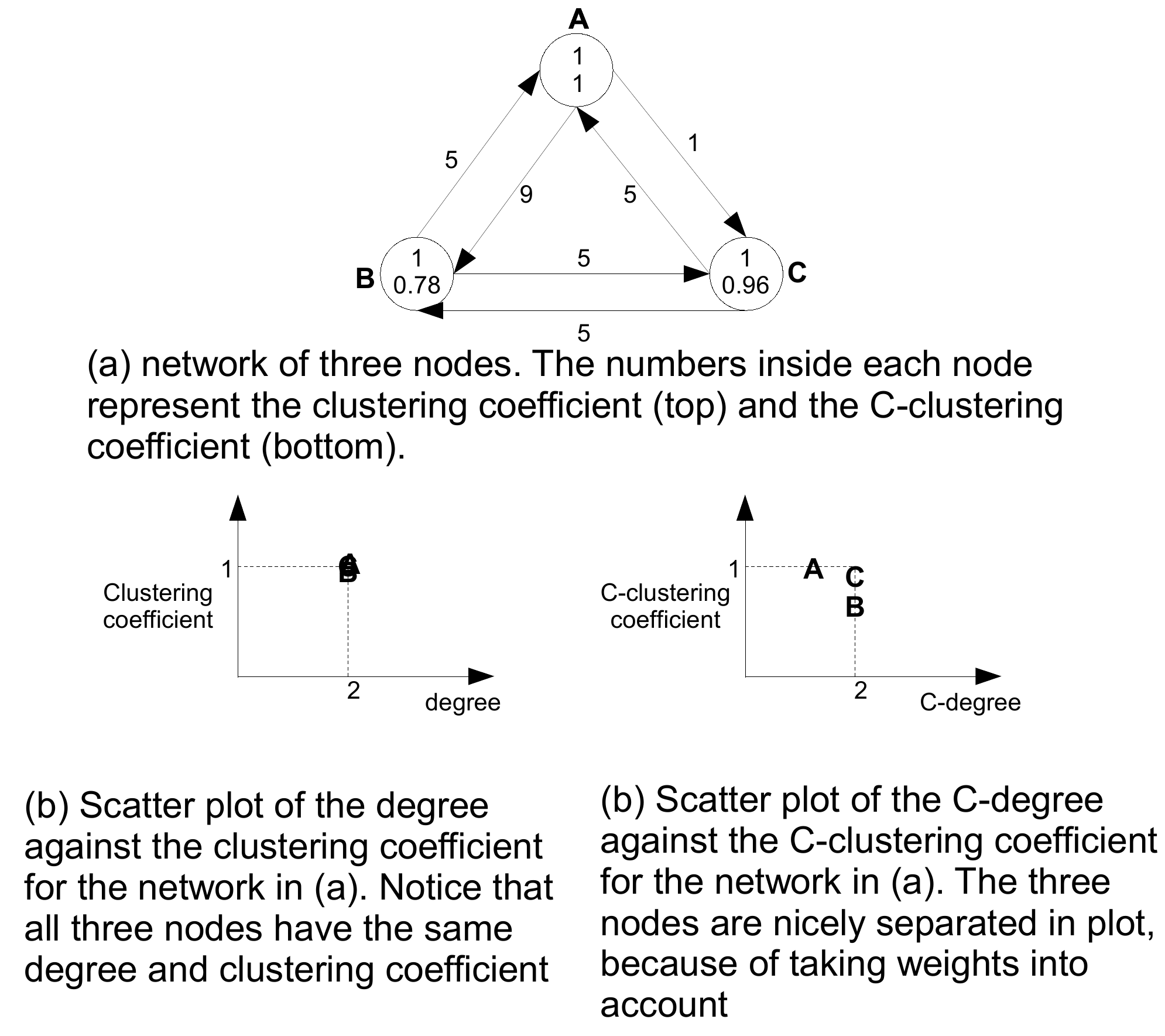}
\caption{{\small Example weighted network of three nodes, comparing the (discrete) clustering coefficient against the C-clustering coefficient. The scatter plot of the degree against the clustering coefficient illustrates the benefit of taking weights into account in distinguishing nodes.}}
      \label{fig-example-cco}
\end{figure}

\section{Experimental Verification}
\label{sec-experiment}
The previous sections provide theoretical analysis as to why our method maintains a connection to the original measures, captures the focus of interaction, and translates that to a continuous spectrum of values. However, several important questions remain unanswered: in a realistic large network, will the proposed generalization be of value? Will it provide more information than the original unweighted measures? Will the generalized measures still maintain connection to their original unweighted measures?

We conduct three types of experimental evaluations (in addition to the theoretical analysis provided earlier) to show the potential of our method: analyzing snapshots of real-world networks (similar to most of the previous work\cite{ahnert07pre,opsahl09sn,opsahl10sn}), analyzing the evolution of a semi-realistic dynamic weighted network, and analyzing the informativeness of our generalization when used to predict labels of network nodes.


\subsection{Analysis of Network Snapshots}
We have analyzed four real world weighted networks\footnote{Available through http://www-personal.umich.edu/~mejn/netdata/}
that capture coauthorships between scientists. Three of which were extracted from preprints on the E-Print Archive \cite{newman01pnas}: condensed matter (an updated version of the original dataset that includes data between Jan 1, 1995 and March 31, 2005), astrophysics, and high-energy theory. The fourth network represents coauthorship of scientists in network theory and experiment \cite{newman06phr}. The weight between two scientists $i$ and $j$ reflects the strength of their collaboration and is given by the equation $w_{ij}=\sum_m{\frac{\delta_i^m \delta_j^m}{n_m-1}}$, where $\delta_i^m=1$ if scientist $i$ was a co-author of paper $m$ and $n_m$ is the number of co-authors for paper $m$\cite{newman01pre}. Table \ref{tab-datasets2} summarize some statistics about the datasets.

\begin{table}
\centering
\begin{tabular}{|r|r|r|r|r|}
\hline
Dataset key& astro-ph & netscience & cond-mat & hep-th\\
\hline
num of nodes & 16706 & 1589 & 40421 & 8361\\
\hline
num of edges & 121251 & 2742 & 175693 & 15751\\
\hline
Weights STD & 0.515 & 0.427 & 0.889 & 1.175\\
\hline
\end{tabular}
\caption{Statistics of the network datasets}
\label{tab-datasets2}
\end{table}

Figure \ref{fig-power} displays the C-degree distribution (CDD) and the (discrete) degree distribution (DD) for the four collaboration network. The figure uses log-log scale with the power law fit based on \cite{clauset07arxiv}.\footnote{Source code adopted from http://www.santafe.edu/\~aaronc/powerlaws/} The CDD follows a pattern similar to the DD, despite taking weights into account. However, the power-law fit for the CDD has steeper decline (higher $\alpha$) than the DD. 
As the degree of a scientist increases, the scientist interacts with a smaller subset of neighbors. While this observation is expected, it raises an interesting question: does the size expected of this collaboration subset remain stable? In other words, on average, does a highly connected scientist collaborates primarily with $X$ number of other scientists, regardless of her degree? or will the number $X$ be a function of the scientist's degree? To answer this question, we define the \emph{degree utilization} metric as the ratio between the C-degree and the degree of a node: $u(v)=\frac{r(v)}{k(v)}$. The degree utilization measures the focus of interaction as a percentage of the original degree. Figure \ref{fig-cdpd} plots the degree utilization against the (discrete) degree for the four collaboration networks. A common pattern emerges in the four networks. For low degrees, the degree utilization is relatively high: scientists with few collaborators rarely focus on subset of these collaborators.  
For nodes with degree greater than some constant the collaboration becomes focused and the degree utilization drops. However, and to our surprise, 
a \emph{cone} is observed, which starts wide at low degrees and gets narrower as the degree increases (the average degree utilization is plotted as a line in the figure). A scientist focuses on a number of collaborators that, on average, is a percentage of the total number of collaborators. This percentage differs from one discipline to another.

\begin{figure}[htbp]
\centering
\includegraphics[width=5in]{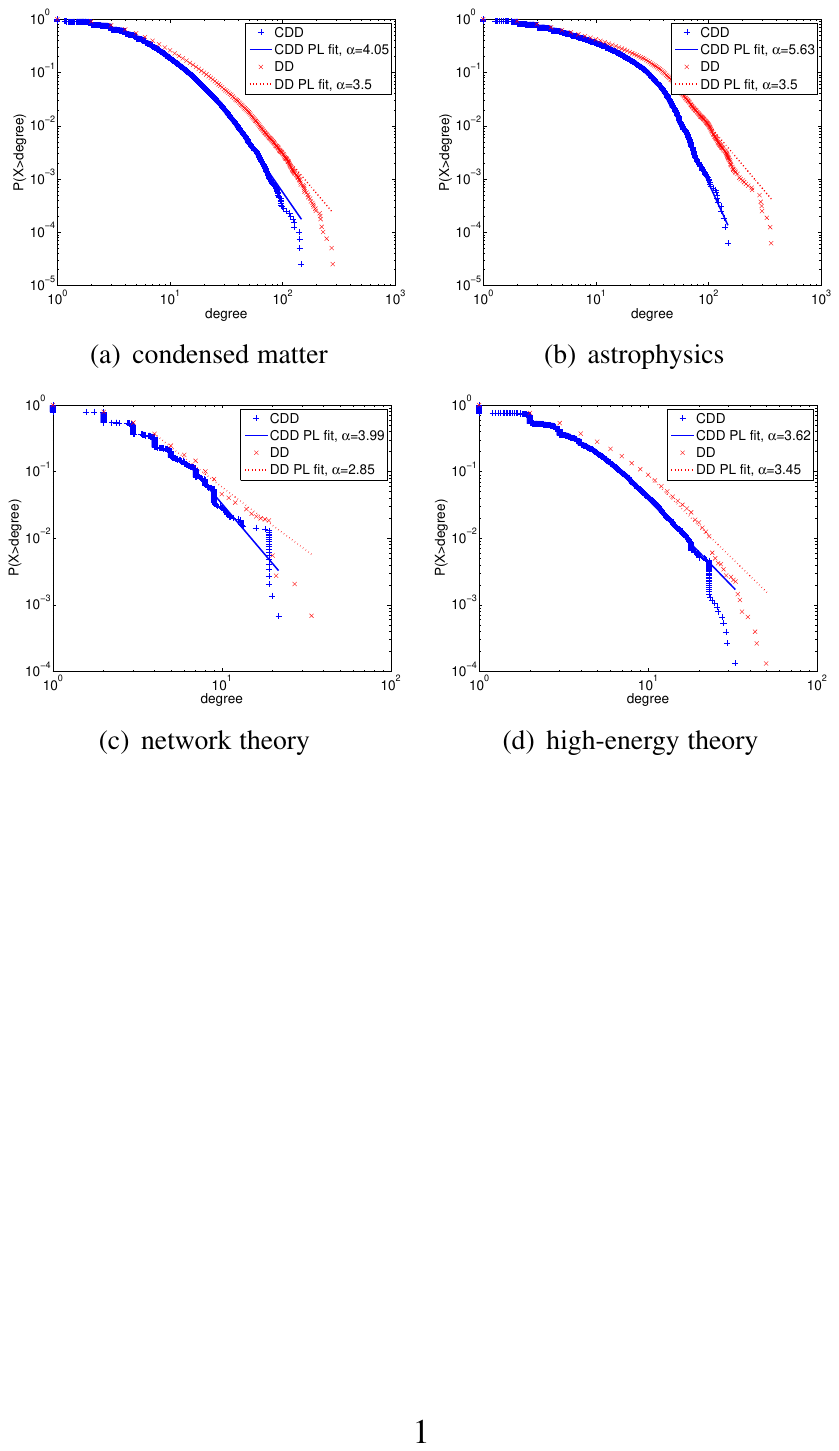}
\caption{{\small Comparing the discrete degree distribution (DD) with the continuous degree distribution (CDD) for the four collaboration networks. The power law fit (PL fit) is also shown with the associated power.}}
      \label{fig-power}
\end{figure}

\begin{figure}[htbp]
\centering
  \includegraphics[width=5in]{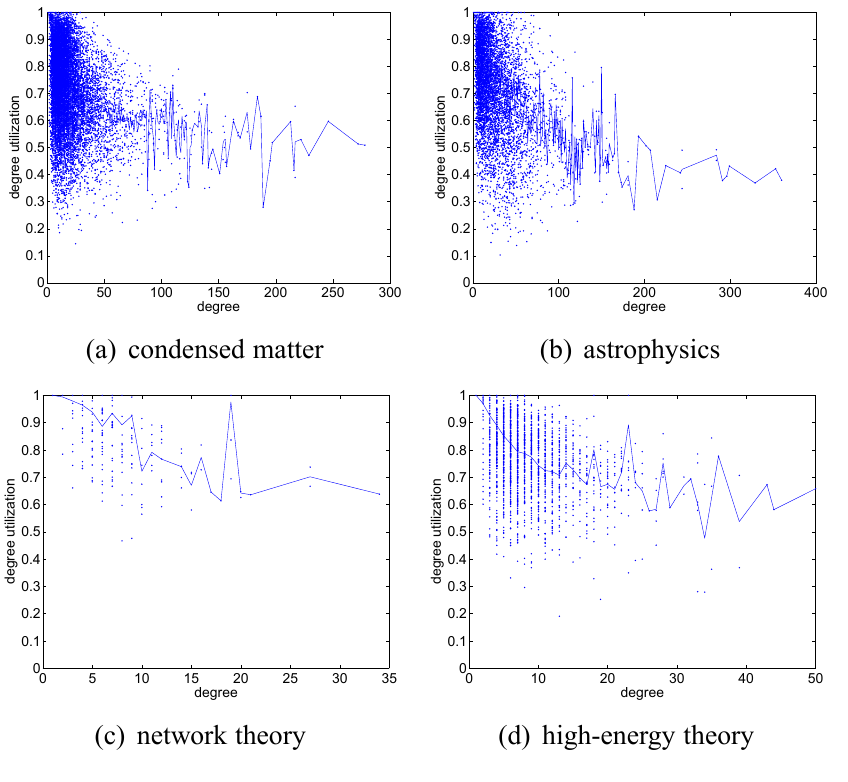}
\caption{{\small Scatter plot of a node degree against its degree utilization for the four collaboration networks. the average utilization per degree is also plotted.}}
      \label{fig-cdpd}
\end{figure}

Figure \ref{fig-cluster-disc} shows the scatter plot of the (discrete) clustering coefficient versus the (discrete) degree (shown in log scale) for the four collaboration networks. The main observation clear from the graph is that in general, the clustering coefficient decreases with the increase of the degree. In other words, the higher the number of collaborators of a scientist, the lesser the density of edges between these collaborators.

\begin{figure}[htbp]
\centering
  \includegraphics[width=5in]{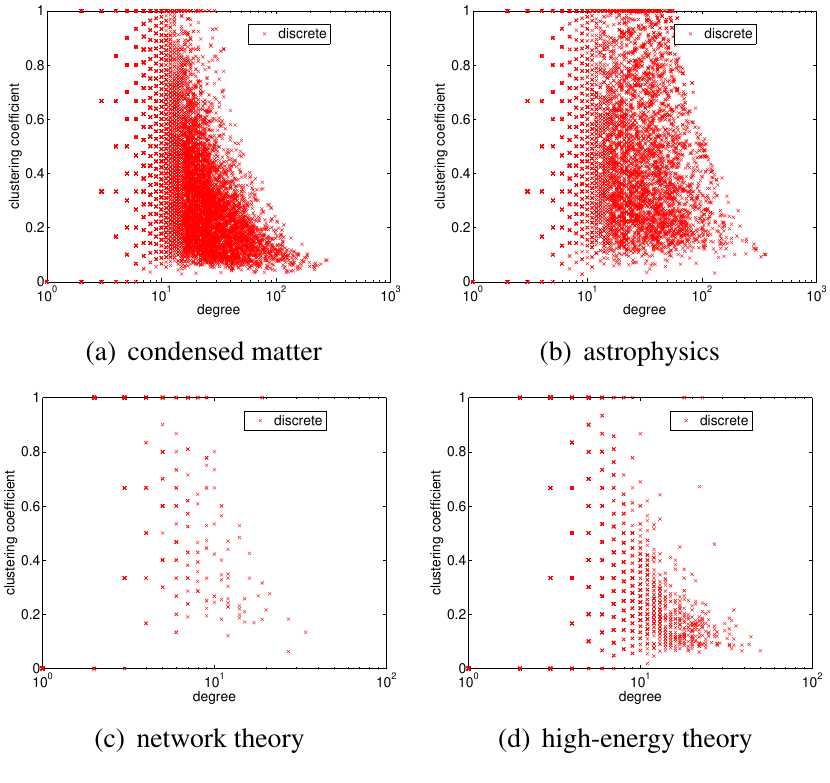}
\caption{{\small Scatter plot of a node's discrete degree against its discrete clustering coefficient for the four collaboration networks.}}
      \label{fig-cluster-disc}
\end{figure}

Figure \ref{fig-cluster-cont} shows the scatter plot of the C-clustering coefficient versus the C-degree for the four collaboration networks. The continuous version of the scatter plot follows the general observation in the discrete case: the clustering coefficient decreases with the increase of the degree. Nevertheless, the scatter plot for the continuous measures covers more area, because both the C-degree and the C-clustering coefficient produce continuous spectrum of values (unlike the discrete degree and the discrete clustering coefficient). Note also that due to maximum cardinality property, Figure \ref{fig-cluster-cont} is shifted towards the origin when compared to Figure \ref{fig-cluster-disc}. More importantly, one can observe an interesting pattern in the continuous scatter plot: nodes with high C-clustering coefficient (above 0.8) tend to have more discrete C-degree. This is clear from the concentration of points with high C-clustering coefficient around the discrete degrees. The same observation is not apparent in points with low clustering coefficient. Using Lemma \ref{lm-cdmax}, this observation means that nodes with incident weights that are more uniform (hence the more discrete degree) tend to cluster with nodes that have more uniform weights among themselves (hence the higher clustering coefficient). We believe this observation reflects research groups: scientists forming cliques of collaboration with almost equal weights over edges (most publications are co-authored by research group members).

\begin{figure}[htbp]
\centering
  \includegraphics[width=5in]{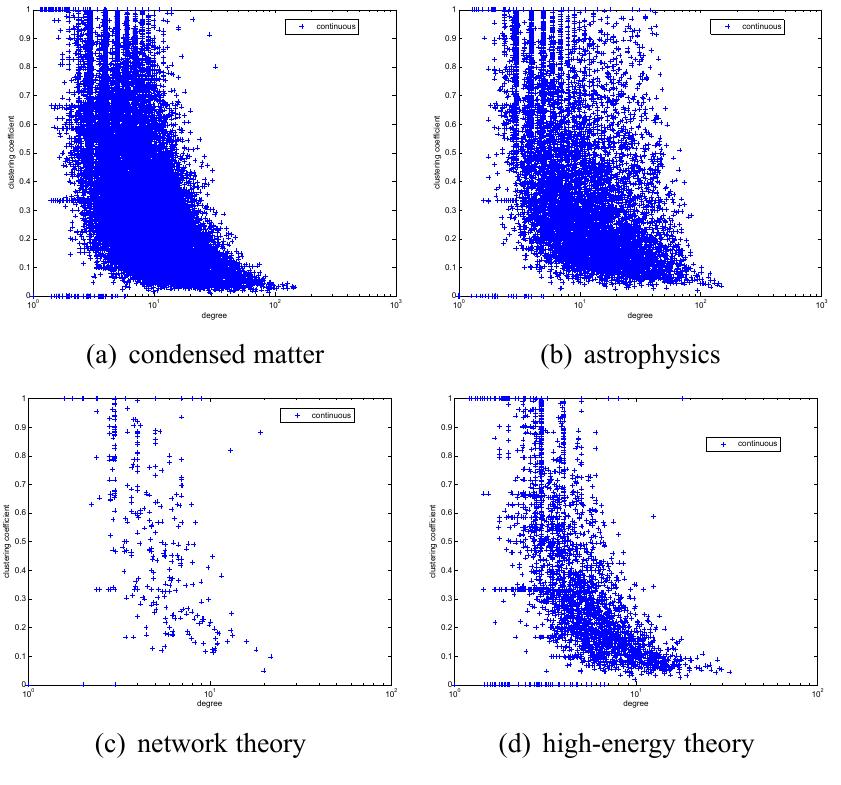}
\caption{{\small Scatter plot of a node's continuous degree against its continuous clustering coefficient for the four collaboration networks.}}
      \label{fig-cluster-cont}
\end{figure}

\subsection{Analysis Using Node Classification}

Recently, researchers discovered that using label-independent network measures can provide useful information in classifying network nodes \cite{gallagher09snakdd}. 
We follow this direction in this section. First, we compute for each node 5 label-independent features: the (discrete) degree, the C-degree, the strength, the clustering coefficient, and the C-clustering coefficient. Then we apply different feature selection algorithms to assess the importance of different features. Each feature selection algorithm selects a subset of features based on some criteria. The criteria differ from one feature selection algorithm to another, but usually takes into account the correlation between a feature (or subset of features) and the node's label.  

We studied 6 labeled datasets, 4 of which represent university websites (University of Texas, Cornell University, Washington University, and University of Wisconsin), while the remaining two represent relationships between industrial companies extracted from news articles (according to two studies).\footnote{The datasets are publicly available at \url{http://netkit-srl.sourceforge.net/data.html}. In a university network, a node represents a web page, which has a label indicating its type (personal web page, department, etc.). A link from one node to another (directed) means there is at least one URL link from the first node to the other. The weight on the link represents the number of such URLs. In an industry dataset, a node represents a company, which has a label indicating its type (transportation, technology, etc.). A link between two nodes exists if the two companies appear in the same news article. The weight represents how many articles the two companies appeared in.}
Table \ref{tab-datasets1} shows some statistics about the 6 datasets: number of nodes, number of links, the percentage of weights that are equal one,\footnote{The percentage of weights that equal one captures the variation in weights more accurately than the standard deviation, which is sensitive to outliers.} and the standard deviation of weights. From the table, we can see that majority of weights equal one for all datasets (ranging from 60\% to almost 80\%), but the Washing. and the Ind-yh datasets have significantly lower percentage (and consequently more variety in weights). 

\begin{table}
\centering
\begin{tabular}{|r|r|r|r|r|r|r|}
\hline
 Data set key & Texas & Cornell & Washing. & Wiscon. & Ind-pr & Ind-yh\\
\hline
num of nodes & 338 & 351 & 434 & 354 & 2189 & 1798\\
\hline
num of edges & 32988 & 2683 & 30462 & 33250 & 13062 & 14165\\
\hline
\% of weights=1 & 0.74 & 0.76 & 0.6 & 0.79 & 0.75 & 0.64\\
\hline
Weights STD & 1.47 & 1.34 & 2.93 & 3.18 & 2.46 & 10.81\\
\hline
\end{tabular}
\caption{Statistics of the labeled network datasets}
\label{tab-datasets1}
\end{table}

As mentioned earlier, 5 features are computed for each node: the degree (D), the clustering coefficient (C), the C-degree (CD), the C-clustering coefficient (CC), and the strength (S). The strength is added here to capture the impact of scale, which is ignored by the first 4 features. Given these features, the idea is to apply feature selection algorithms and observe for each dataset which features are considered significant or more important than others. Since different algorithms may select different features, more than one algorithm were applied. 
We focus here on publicly available algorithms through the WEKA software package, using the default settings \cite{weka09}. We also avoid algorithms that discretize features, because discretizing a continuous generalization diminishes its advantage. 

Table \ref{tab-cfs} shows the outcome of applying the first feature selection algorithm: CfsSubsetEval \cite{weka09}. The algorithm outputs a subset of features that have high correlation with the class and low correlation among themselves. As the table shows, the C-degree was selected in 5 of the 6 datasets, compared to only 4 datasets where the degree was selected. It is interesting to note here that the only dataset where the C-degree was not selected (Ind-pr), the degree was not selected as well. On the other hand, the strength was selected in that dataset, which means that for Ind-pr dataset and the CfsSubsetEval algorithm, the class is more correlated with the scale rather than the degree of focus. It is interesting to observe here that although the strength measure takes weights into account, the degree is still more informative. This observation explains why the degree is still widely used in analyzing weighted networks. Also the C-clustering coefficient was selected in 3 datasets, while the original clustering coefficient was selected only once. 

\begin{table}
\centering
\begin{tabular}{|r|r|r|r|r|r|}
\hline
Texas & Cornell & Washing & Wiscon & Ind-pr & Ind-yh\\
\hline
CD & CD & CD & CD & CC & CD\\
\hline
D & D & D & D & C & \\
\hline
CC & C & CC &  & S & \\
\hline
 & S &  &  &  & \\
\hline
\end{tabular}
\caption{Selected features using CfsSubsetEval}
\label{tab-cfs}
\end{table}

Table \ref{tab-ib1} shows the results of another feature selection algorithm: ClassifierSubsetEval with 1-nearest neighbor classifier. The algorithm evaluates subsets of attributes using the accuracy of the 1-nearest neighbor classifier to prefer one subset over the other. As shown in the table, this feature selection algorithm confirms a clear advantage to our approach. None of the original unweighted measures was selected for any dataset. Unlike the previous feature selection algorithm, however, the C-clustering seems to be more dominant across datasets (selected in all of them). Interestingly, the strength was selected in 5 out of the 6 datasets, complementing the disparity sensitivity with the scale sensitivity for this particular feature selection algorithm.

\begin{table}
\centering
\begin{tabular}{|r|r|r|r|r|r|}
\hline
Texas & Cornell & Washing & Wiscon & Ind-pr & Ind-yh\\
\hline
CC & CD & CD & CD & CD & CC\\
\hline
S & CC & CC & CC & CC & S\\
\hline
 & S & S &  & S & \\
\hline
\end{tabular}
\caption{Selected features subset using ClassifierSubsetEval with 1-nearest neighbor classifier}
\label{tab-ib1}
\end{table}

Table \ref{tab-relief} shows the numerical ranking for the features using the ReliefFAttributeEval algorithm. The algorithm evaluates each attribute by repeatedly sampling an instance and considering the value of the attribute for the nearest instances of the different classes. To simplify comparison, the values for each dataset was divided by the maximum value of the column (higher is better). From the table we can see that the C-degree always has higher value than the original degree. It is interesting to note that the difference is highest (Ind-yh and Washing) in datasets with less percentage of weights equal 1 (Table \ref{tab-datasets1}). The C-clustering coefficient was better than the traditional clustering coefficient in 4 datasets, and worse in 2 datasets. Interestingly, the strength has the highest value in the Ind-pr dataset, which is consistent with the results of the previous feature selection algorithm (as discussed earlier regarding Table \ref{tab-ib1}, Ind-pr is the only dataset where neither the C-degree not the degree was selected).

\begin{table}
\centering
\begin{tabular}{|r|r|r|r|r|r|r|}
\hline
Feature& Texas & Cornell & Washing & Wiscon & Ind-pr & Ind-yh\\
\hline
CD & 1 & 1 & 0.884 & 1 & 0.808 & 1\\
\hline
D & 0.925 & 0.999 & 0.780 & 0.989 & 0.807 & 0.588\\
\hline
CC & 0.633 & 0.503 & 1 & 0.297 & 0.397 & 0.553\\
\hline
C & 0.547 & 0.480 & 0.535 & 0.145 & 0.465 & 0.752\\
\hline
S & 0.354 & 0.395 & 0.473 & 0.246 & 1 & 0.248\\
\hline
\end{tabular}
\caption{Selected features subset using the ReliefFAttributeEval algorithm}
\label{tab-relief}
\end{table}

\subsection{Analysis of Dynamic Networks}

We believe our methodology will be most effective in studying the evolution of weighted networks over time. Unfortunately, we were not able to find any publicly available dataset that provided such an evolution. Most of the available datasets of weighted networks provided only snapshots of a particular network at a particular point in time. 
The other possibility is to synthesize the evolution of a weighted network. Several models were proposed\cite{li07pa}. However, all these models were designed to capture the properties that are known (e.g. power-law of degree distribution and/or the power-law of the strength distribution) and are less helpful in evaluating new measures. We opted for a simpler alternative: growing the weights of some initial network to eventually be equivalent to a publicly-available network snapshot. More formally: let $N_{last}$ be some snapshot of a real world network (e.g. the science collaboration network). Let $N_{0}$ be some initial network that has the same structure as $N_{last}$ but with different initial weights. We then developed two simple mechanisms to evolve a network from $N_0$ to $N_{last}$.

\begin{itemize}
\item Assuming initially all weights are zero, increase the weight of every edge proportional to the edge's weight in $N_{last}$.
\item Assuming initially the interaction for each node is not focused (i.e. all weights are equal), while maintaining the same strength, change the weights gradually in the direction of $N_{last}$.
\end{itemize}

Figure \ref{fig-dynamic-netscience} shows the evolution of the our generalized degree distribution for each mechanism using the netscience collaboration network. As expected, our generalization captures the change in the focus of interaction properly, while other measures fail. When there is no change in the focus, our generalization remains unchanged.

    \begin{figure}
    \centering
		\subfigure[Changing scale]{\includegraphics[width=2.25in]{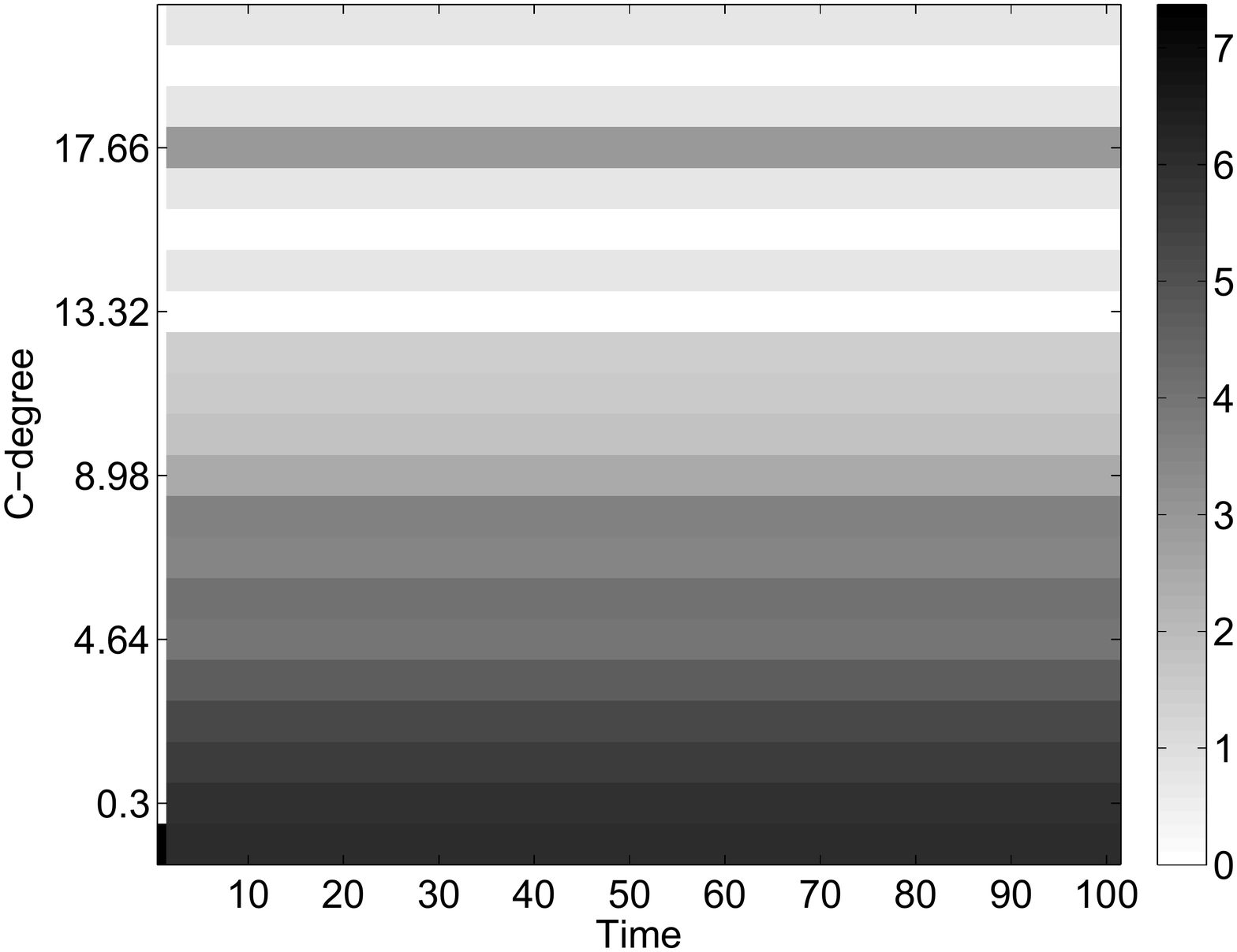}}
		\subfigure[Changing focus]{\includegraphics[width=2.25in]{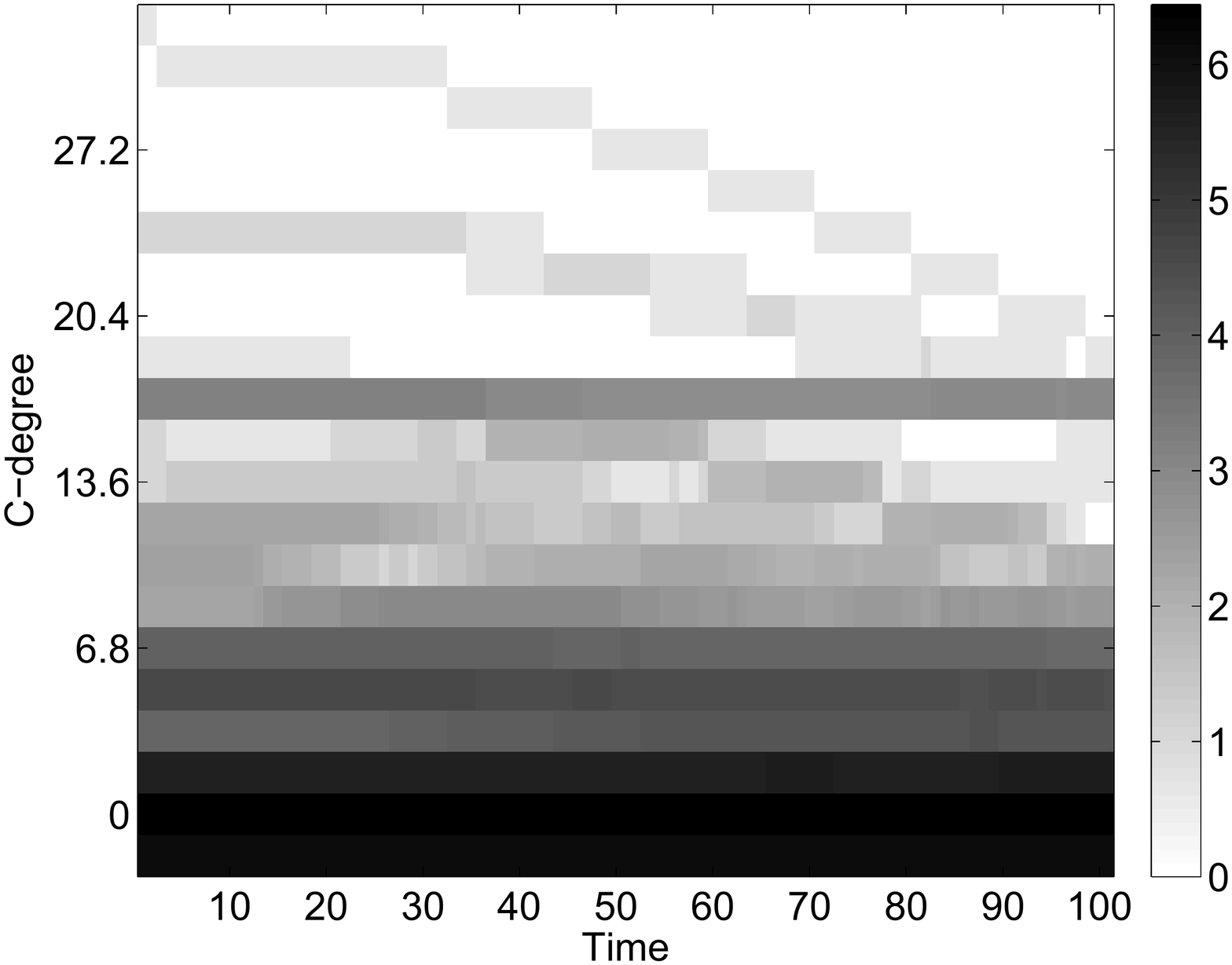}}
		\caption{\small The evolution of the netscience networks under two simple mechanisms: changing the scale and changing the focus. As expected, our generalization captures the change in the focus properly, while other measures fail. When there is no change in the focus, our generalization remains unchanged. Note that the frequency (the gray scale) is in log scale.}
      \label{fig-dynamic-netscience}
    \end{figure}

\section{Related Work}
\label{sec-related}

In general, one can classify weighted network measures into two classes: measures that generalize unweighted network measures to take weights into account, and measures that have no connection to unweighted measures. Surveying all weighted measures that have no connection to unweighted measures is beyond the scope of this paper and have little relevance to the contribution of this paper, which generalizes unweighted network measures. For completeness, we provide here a sample of these measures that are related to some unweighted measure (interested reader may refer to survey papers on the subject \cite{newman04preA,barth05pa,boccaletti06pr}). 
The \emph{strength} of a node is the summation of all weights incident to a node. The strength becomes identical to the node's degree in the very special case when all the weights are equal to 1, but it has very weak partial ordering among nodes. For example, all the nodes in Figure \ref{fig-cdd} have the same strength of 1. 
The \emph{weight distribution} is similar to the degree distribution except that it measures the frequency of a particular edge weight. 
A more recent work \cite{mcglohon08kdd} analyzed a graph's total weight, $\sum_{e\in E} w(e)$, against the graph's total number of edges, $|E|$, over time. That work also analyzed the degree of a node, $k(v)$, against the node's strength, $s(v)$. While useful, the above measures neither captured the degree of focus in interactions nor provided a methodology for generalizing unweighted measures, unlike our proposed generalization.
The network measure $Y(v)=\sum_{e \in E(v)} \left(\frac{w(e)}{s(v)}\right)^2$ successfully captured the disparity of interaction within edges incident to node $v$ \cite{almaas04n}.
However, unlike our work here, the $Y$ measure is not a \emph{generalization} of the degree measure. 


There have been several attempts to generalize specific unweighted measures. The weighted clustering coefficient \cite{barrat04pnas} was an attempt to generalize the clustering coefficient. The generalization relied on an alternative definition of the clustering coefficient that used triplets \cite{watts98n}. A triplet connected to a node is a subgraph containing the original node in addition to two other connected neighbors. The intuition behind the weighted clustering coefficient for node $i$ is to weigh every edge between two of its neighbors, $j$ and $k$, using the weights on edges $(i,j)$ and $(i,k)$. However, unlike our generalization, the weight on edge $j,k$ was ignored. 
A recent attempt to generalize the clustering coefficient used the ratio between the total \emph{value} of closed triplets and the total value of all triplets \cite{opsahl09sn}. The authors proposed four functions to evaluate (summarize) weighted triplets: the arithmetic mean, the geometric mean, the minimum, and the maximum. The four proposed functions (and therefore the proposed generalization) have weaker distinguishing powers than our proposed method. For example, all the nodes in Figure \ref{fig-example-cco} will have a generalized clustering coefficient of 1 if the method in \cite{opsahl09sn} is used (this limitation was previously reported \cite{opsahl09sn}). On the other hand, our proposed generalized clustering coefficient successfully distinguishes all the three nodes.

The ensemble approach \cite{ahnert07pre} provides a methodology for generalizing almost all unweighted network measures. The first step of the method was to normalize edge weights to ensure all weights are between 0 and 1 (more restrictive than our approach, which only assumes weights are non-negative). The next step was to randomly generate an ensemble of unweighted networks from the original weighted network, where the weight of an edge represented the probability of generating the edge. The final step was to compute the generalized unweighted measure as the average of the unweighted measure for each network in the ensemble. Because the ensemble approach relies on computing the expectation, the method could not provide any partial ordering guarantee. For example, suppose two nodes have the following sets of incident edge weights $A=\{9,8,2,1\}$ and $B=\{9,5,5,1\}$. Under the ensemble approach, both nodes will have the same generalized degree. Using our proposed generalization, node $B$'s degree is guaranteed to be greater than node $A$'s C-degree by Lemma \ref{lemma-po}.
Another side-effect of relying on the expectation is the limited connection to the original unweighted measures. Only if the weights were normalized such that all weights equaled exactly 1 would the ensemble approach provide generalizations equal to the original unweighted measures.

The most recent attempt to generalize the degree (along with betweenness and closeness) relied on mixing total weights (strength) with the cardinality (degree) \cite{opsahl10sn}. The generalized degree based on this method is the $\alpha$-degree that we described earlier: $\alpha$-degree = $s^\alpha k^{1-\alpha}$. Although the generalization methodology were applied to the betweenness and closeness centrality measures, only the degree generalization was evaluated. More importantly, the generalization has a tuning parameter without clear guidelines regarding how to set it, unlike our methodology which is parameterless. Also the connection to the original unweighted measure does not depend on the weights, but on the tuning parameter (when $\alpha$ equals 0). As a result, the $\alpha$-generalization is not really a generalization, but rather a new class of weighted measures.

\section{Conclusion}
\label{sec-conclude}
We proposed a new methodology for generalizing measures of unweighted networks. Our method captures the degree of focus in the interaction over edges, and reduces to the original unweighted measure if the interactions are not focused. We illustrated the applicability of our method by generalizing four unweighted network measures including the degree and the clustering coefficient. We complemented the theoretical analysis of our proposed generalization with the experimental analysis of 3 types of datasets. The analysis of labeled dataset classification showed how feature selection algorithms prefer our generalized measures to the original unweighted measures, across different datasets and different feature selection algorithms. The analysis of un-labeled datasets (collaboration networks) showed that the generalized degree distribution follows a similar pattern to the traditional degree distribution, but with steeper decline (larger exponent of the power-law fit). The analysis exposed scientific research groups and showed that, in expectation, scientists focus their collaboration on fixed percentage, regardless of the number of collaborators. 

Due to the large body of research that relied on unweighted measures for analyzing networks, there are several interesting directions for following up on this work. We are currently investigating the use of the C-degree in simple network navigation to replace the degree. We are also investigating the use of the generalized measures to study the evolution of interaction networks.

\bibliographystyle{spmpsci}      

\section*{Appendix A: Proofs of Effective Cardinality Properties}
\begin{theorem}
The effective cardinality satisfies the three properties described above: the maximum cardinality, the minimum cardinality, and the consistent partial ordering.
\end{theorem}
\begin{proof}
The proof follows from the following three lemmas.
\end{proof}

\begin{lemma}
\label{lm-cdmax}
The effective cardinality satisfies the maximum cardinality property.
\end{lemma}

\begin{proof}
When all the weights are equal to a constant $C$ we have
\[
\forall e \in E': \frac{w(e)}{\sum_{o\in E'}w(o)}=\frac{C}{C|E'|}=\frac{1}{|E'|}
\] 
We then have
\begin{align*}
c(E') & = 2^{\sum_{e \in E'} \frac{1}{|E'|} \log_2(|E'|)} \\
& = 2^{\log_2(|E'|)} \\
& = |E'|
\end{align*}

In other words, both the cardinality and the effective cardinality of a weighted set of edges become equivalent when the weights are uniform. The effective cardinality is also maximum in this case, because the exponent is the entropy of the weight probability distribution, which is maximum when weights are uniform over edges. 
\end{proof}

\begin{lemma}
The effective cardinality satisfies the minimum cardinality property.
\end{lemma}

\begin{proof}
When the set of edges is empty, then the effective cardinality is zero by definition.
When all weights are zero except only one weight that is greater than zero, then weight probability distribution is deterministic and the entropy is zero, therefore the effective cardinality will be 1.
\end{proof}

\begin{lemma}
\label{lemma-po}
The effective cardinality satisfies the consistent partial order property.
\end{lemma}

\begin{proof}
Let $E'_1$ and $E'_2$ be two (edge) sets such that $|E'_1|=|E'_2|$ (both have the same cardinality). Let $W_1$ and $W_2$ be the corresponding sets of weights, where $\sum_{e1 \in E'_1} w(e1)=\sum_{e2 \in E'_2} w(e2) = S$ (the total weights are equal). 
Furthermore, let $|W_1 \bigcap W_2|=n-2, \{w_{11},w_{12}\} = W_1 - W_2,\{w_{21},w_{22}\} = W_2 - W_1$, where the '$-$' operator is the "set difference" operator (the two sets share the same weights except for two elements in each set), and $|w_{11}-w_{12}| < |w_{21}-w_{22}|$ (the weights of $W_1$ are more uniform than the weights of $W_2$). To prove that the effective cardinality satisfies the consistent partial ordering property, we need to prove that $c(E'_1)>c(E'_2)$.

Without loss of generality, we can assume that $w_{11} \ge w_{12}$ and $w_{21} \ge w_{22}$, therefore $w_{11}-w_{12} < w_{21}-w_{22}$. We then have 
\[
w_{11}+w_{12}=S - \sum_{w \in W_1 \bigcap W_2}w = w_{21}+w_{22}
\]
or
\[
\frac{w_{11}+w_{12}}{S}=1 - \sum_{w \in W_1 \bigcap W_2}\frac{w}{S} = \frac{w_{21}+w_{22}}{S} = L
\]
therefore 
\[
L \ge \frac{w_{21}}{S} > \frac{w_{11}}{S} \ge \frac{L}{2} \ge L-\frac{w_{11}}{S} > L-\frac{w_{21}}{S}
\] 
where $\frac{w_{12}}{S} = L-\frac{w_{11}}{S}$ and $\frac{w_{22}}{S} = L-\frac{w_{21}}{S}$. Then from Lemma \ref{lemma-ent} we have $h(L,\frac{w_{11}}{S}) > h(L,\frac{w_{21}}{S})$, or
\begin{align*}
-\frac{w_{11}}{S}lg(\frac{w_{11}}{S}) - (L-\frac{w_{11}}{S})lg(c-\frac{w_{11}}{S}) > \\ -\frac{w_{21}}{S}lg(\frac{w_{21}}{S}) - (L-\frac{w_{21}}{S})lg(c-\frac{w_{21}}{S})
\end{align*}
Therefore $H(E'_1) > H(E'_2)$, because the rest of the entropy terms (corresponding to $W_1 \bigcap W_2$) are equal, and consequently $c(E'_1)>c(E'_2)$.
\end{proof}

\begin{lemma}
\label{lemma-ent}
The quantity $h(C,x)=-x\lg(x) - (C-x) \lg(C-x)$ is symmetric around and maximized at $x=\frac{C}{2}$ for $C \ge x \ge 0$.
\end{lemma}
\begin{proof}
\[
h(C,\frac{C}{2}+\delta)=-(\frac{C}{2}+\delta)\lg(\frac{C}{2}+\delta) - (\frac{C}{2}-\delta) \lg (\frac{C}{2}-\delta) = h(C,\frac{C}{2}-\delta)
\]
Therefore $h(C,x)$ is symmetric around $c/2$.
Furthermore, $h(C,x)$ is maximized when 
\[
\frac{\partial h(C,x)}{\partial x} = 0 = -1 -\lg x + 1 + \lg(C-x)
\]
or
\[
\lg x = \lg(C-x)
\]
Therefore $h(C,x)$ is maximized at $x=C-x = \frac{C}{2}$.
\end{proof}


\end{document}